\documentclass[11pt]{article}
\RequirePackage{algorithm,algorithmic,amssymb,epsf,epic,enumerate,enumitem,graphicx,url,lineno}
\RequirePackage{fullpage}
\RequirePackage{amsmath,amssymb,amsthm}

\RequirePackage{ifpdf}
\ifpdf    
\RequirePackage{hyperref}
\else    
\RequirePackage[hypertex]{hyperref}
\fi

\newcommand{\eps}{\epsilon}
\DeclareMathOperator{\ddim}{ddim}
\DeclareMathOperator{\MST}{MST}

\def\denseformat{
\setlength{\textheight}{9.5in}
\setlength{\textwidth}{6.9in}
\setlength{\evensidemargin}{-0.3in}
\setlength{\oddsidemargin}{-0.3in}
\setlength{\headsep}{10pt}
\setlength{\topmargin}{-0.44in}
\setlength{\columnsep}{0.375in}
\setlength{\itemsep}{0pt}
}

\newtheorem{theorem}{Theorem}[section]

\newtheorem{lemma}[theorem]{Lemma}

\newtheorem{corollary}[theorem]{Corollary}
\newtheorem{fact}[theorem]{Fact}
\newtheorem{comment}[theorem]{Comment}

\newtheorem{observation}[theorem]{Observation}

\def\boldhead#1:{\par\vskip 7pt\noindent{\bf #1:}\hskip 10pt}
\def\ithead#1:{\par\vskip 7pt\noindent{\it #1:}\hskip 10pt}
\def\ceil#1{\lceil #1\rceil}

\def\inline#1:{\par\vskip 7pt\noindent{\bf #1:}\hskip 10pt}
\def\midinline#1:{\par\noindent{\bf #1:}\hskip 10pt}
\def\dnsinline#1:{\par\vskip -7pt\noindent{\bf #1:}\hskip 10pt}
\def\ddnsinline#1:{\newline{\bf #1:}\hskip 10pt}
\def\largeinline#1:{\par\vskip 7pt\noindent{\large\bf #1:}\hskip 10pt}
\long\def\comment #1\commentend{}
\long\def\commhide #1\commhideend{}
\long\def\commfull #1\commend{#1}
\long\def\commabs #1\commenda{}
\long\def\commtim #1\commendt{#1}
\long\def\commb #1\commbend{}
\long\def\commedit #1\commeditend{} 

\long\def\commB #1\commBend{}       

\long\def\commex #1\commexend{}     

\long\def\commsiena #1\commsienaend{}  
                                         
\long\def\commBI #1\commBIend{}  
                                         
\long\def\CProof #1\CQED{}

\def\blackslug{\hbox{\hskip 1pt \vrule width 4pt height 8pt
    depth 1.5pt \hskip 1pt}}
\def\QED{\quad\blackslug\lower 8.5pt\null\par}
\def\inQED{\quad\quad\blackslug}

\long\def\PPP#1{\noindent{\bf Proof:}{ #1}{\quad\blackslug\lower 8.5pt\null}}

\long\def\denspar #1\densend
{#1}

\newif\ifnotesw\noteswtrue
   {\ifnotesw\marginpar[\hfill\(\top\)]{\(\top\)}\fi}%
   {\ifnotesw\marginpar[\hfill\(\bot\)]{\(\bot\)}\fi}

\newcommand{\mnote}[1]%
    {\ifnotesw\marginpar%
        [{\scriptsize\it\begin{minipage}[t]{\marginparwidth}
        \raggedleft#1%
                        \end{minipage}}]%
        {\scriptsize\it\begin{minipage}[t]{\marginparwidth}
        \raggedright#1%
                        \end{minipage}}%
    \fi}

\def\MathF{\hbox{\rm I\kern-2pt F}}
\def\MathP{\hbox{\rm I\kern-2pt P}}
\def\MathR{\hbox{\rm I\kern-2pt R}}
\def\MathZ{\hbox{\sf Z\kern-4pt Z}}
\def\MathN{\hbox{\rm I\kern-2pt I\kern-3.1pt N}}
\def\MathC{\hbox{\rm \kern0.7pt\raise0.8pt\hbox{\footnotesize I}
\kern-4.2pt C}}
\def\MathQ{\hbox{\rm I\kern-6pt Q}}

\newsavebox{\ttop}\newsavebox{\bbot}

\def\eps{\epsilon}

\denseformat

\newcommand {\ignore} [1] {}

\begin{document}
\linenumbers
\title{Light spanners for Snowflake Metrics}
\author{
Lee-Ad Gottlieb \thanks{Department of Computer Science and Mathematics, Ariel University, Ariel, Israel.
E-mail: {\tt leead@ariel.ac.il}.}
\and
Shay Solomon \thanks{Department of Computer Science and Applied Mathematics, The Weizmann Institute of Science, Rehovot 76100, Israel.
E-mail: {\tt shay.solomon@weizmann.ac.il}.
This work is supported by the Koshland Center for basic Research.}}

\date{\empty}

\begin{titlepage}
\def\thepage{}
\maketitle

\begin{abstract}

A classic result in the study of spanners is the existence of light low-stretch spanners for Euclidean spaces. 
These spanners have arbitrary low stretch, and weight only a constant factor greater than that of the minimum spanning 
tree of the points (with dependence on the stretch and Euclidean dimension). A central open problem in this field asks 
whether other 
spaces admit low weight spanners as well -- for example metric space with low intrinsic dimension -- yet only a handful of 
results of this type are known. 

In this paper, we consider snowflake metric spaces of low intrinsic dimension. The $\alpha$-snowflake of a 
metric $(X,\delta)$ is the metric $(X,\delta^\alpha)$, for $0<\alpha<1$. By utilizing an approach completely 
different than those used for Euclidean spaces, we demonstrate that snowflake metrics admit light spanners. 
Further, we show that the spanner is of diameter $O(\log n)$, a result not possible for Euclidean spaces. As an 
immediate corollary to our spanner, we obtain dramatic improvements in algorithms for the traveling salesman 
problem in this setting, achieving a polynomial-time approximation scheme with near-linear runtime.
Along the way, we also show that all $\ell_p$ spaces admit light spanners, a result of interest in its own
right.
 
\end{abstract} 
\end{titlepage}

\pagenumbering {arabic}

\section{Introduction}

Given a complete graph $G$, a \emph{$(1+\epsilon)$-spanner} for $G$ is a subgraph $H \subset G$ which preserves 
all pairwise distances in $G$ to within a factor of $1+\epsilon$. Low-stretch spanners with additional favorable 
properties -- such as low degree, weight or small hop diameter -- have been the object of much study, in settings such
as Euclidean space, planar graph metrics, and metrics with low intrinsic dimension 
\cite{AS97,AMS94,ADMSS95,ACCDSZ-96,CG06,GR081,GR082,DES08,Sol11,ES13}.

A spanner $H$ is said to be \emph{light} if its weight is proportional to the weight of the minimum spanning 
tree of $G$, $w(\MST(G))$, and its \emph{lightness} is the constant (or term) multiplying $w(\MST(G))$.
A major result of the nineties is that $d$-dimensional Euclidean spaces admit light
$(1+\epsilon)$-spanners, with lightness $\epsilon^{-O(d)}$ \cite{DHN93}. 
An important result in its own right, the light Euclidean spanner is also a central component in the fastest 
polynomial time approximation scheme (PTAS)
for the Euclidean traveling salesman problem (TSP): Using a light spanner, a $(1+\epsilon)$-approximate tour can be
computed in time $2^{\epsilon^{O(d)}} n + 2^{O(d)} n \log n$ \cite{A-98,RS-98}.

The existence proof for light Euclidean spanners is complex. At its core, it relies on the \emph{leapfrog property} 
specific to Euclidean space. It seems difficult to extend this proof to other natural spaces,
and in fact light spanners are known for only a handful of settings. These include 
planar graphs \cite{ADDJS93},
unit disk graphs \cite{KPX-08},
and graphs of bounded pathwidth \cite{GH-12}
and bounded genus \cite{DMM-10}.
In fact, a central conjecture in this area asks whether all metric spaces $M$ with low intrinsic dimension admit light 
$(1+\epsilon)$-spanner. The best lightness bound known
for spanners in these metrics is $\Omega(\log n)$ \cite{Smid09,ES13}.

In this paper we take a step towards this conjecture, by showing that light spanners exist for
snowflake metrics of low doubling dimension; 
the $\alpha$-snowflake of a metric $(X,\delta)$ is the metric $(X,\delta^\alpha)$, with $0 < \alpha < 1$.
Snowflake metrics have been a focus of study in the recent literature \cite{A-83,GKL03,LMN-04,ABN-08,GK-11,BRS-11,NN-12,N-13}.
We will give two separate proofs for the existence of light spanners for snowflake metrics.
As an immediate corollary, we derive a fast approximation algorithm for TSP on snowflake metrics,
which yields a dramatic improvement on what was previously known \cite{BGK12}.

The first proof, presented in Section \ref{sec:first}, 
is based on a new observation: For any constant-dimensional vector space $V$,
the edges of a light $(1+\epsilon)$-spanner for space $(V,\ell_2)$ also form a light
$(1+\eps^{O(1)})$-spanner for any space $(V,\ell_p)$, when $p \ge 1$. 
This result is of independent interest, and also implies an efficient approximation
algorithm for traveling salesman in $\ell_p$ (see Corollary \ref{cor:lptsp}).
Further, combining this result with the fact that the $\alpha$-snowflake of any metric $M$
embeds with $1+\epsilon$ distortion into $\epsilon^{-O(\ddim(M))}$-dimensional $\ell_\infty$
\cite{HM-06}, we conclude in
Theorem \ref{thm:first-main} that every snowflake metric $M$ admits $(1+\epsilon)$-spanners of
lightness $2^{\epsilon^{-O(\ddim(M))}}$.

Unsatisfied with the lightness bound provided by Theorem \ref{thm:first-main}, we present a second proof
in Section \ref{sec:second}. This proof is more involved, but yields an exponentially better
lightness bound. The proof considers the standard net-tree spanner (NTS), and demonstrates
that for snowflake metrics, the NTS is light. 
In Theorem \ref{thm:basic}, we show that the NTS on a snowflake metric $M$ has stretch 
$(1+\epsilon)$ with lightness $\epsilon^{-O(\ddim(M))}$.
The advantage of this spanner over that of Section \ref{sec:first} is threefold:
(1) It bypasses the heavy machinery of the leapfrog property and the use of low-distortion embeddings. 
(2) The dependencies on $\eps$ and the doubling dimension in the lightness bound,
as well as the leading constant therein, are significantly smaller. This is of particular interest to
practitioners in the field.
(3) The NTS is the central tool in several spanner constructions that combine small weight with other favorable properties, 
and the weight bound in all these constructions depends on the weight of the NTS.
Consequently, we improve the weight bound in all these constructions (see Section \ref{sec:apps}).

An interesting property of our spanners is that they have logarithmic hop-diameter,
a result not possible for regular metric spaces. In fact, even 1-dimensional 
Euclidean space requires linear hop-diameter for light spanners \cite{DES08}.

The paper is organized as follows:
We close this introductory section with preliminary notes and definitions.
We present the first proof in Section \ref{sec:first}:
We show that there exists a light spanner for all $\ell_p$ spaces, and that this
implies a light spanner for snowflake metrics as well.
The second proof appears in Section \ref{sec:second}.
Finally, we detail some further applications of our proofs,
including TSP algorithms, in Section \ref{sec:apps}.

\subsection{Preliminaries}

Let $(X,\delta)$ be an arbitrary $n$-point metric.
Without loss of generality, we  assume that the minimum inter-point distance of $X$ is equal to $1$.
We denote by $\Delta = \max_{u, v \in X}\delta(u, v)$  the \emph{diameter} of $X$.

\paragraph{Doubling dimension.}  
The \emph{doubling dimension} of a metric space $(X, \delta)$, denoted by $\ddim(X)$ (or $\ddim$ when the context is 
clear), is the smallest value $\rho$ such that every ball in $X$ can be covered by $2^\rho$ balls of half the radius. A 
metric is called \emph{doubling} if its doubling dimension is bounded above by a constant \cite{GKL03}.

\paragraph{Hierarchical Nets.}  
A set $Y \subseteq X$ is called an \emph{$r$-cover} of $X$ if for any point $x \in X$ there is a point $y \in Y$, with $\delta(x, y) \le r$.
A  set $Y$ is an \emph{$r$-packing} if for any pair of distinct points $y, y' \in Y$, it holds that $\delta(y, y') > r$.
We say that a set $Y \subseteq X$ is an \emph{$r$-net} for $X$ if $Y$ is both an $r$-cover of $X$ and an $r$-packing.
By recursively applying the definition of doubling dimension, we can derive the following key fact \cite{GKL03}. 

\vspace{-0.04in}
\begin{fact} \label{prop:small_net}
Let $R \geq 2r > 0$ and let $Y \subseteq X$ be an $r$-packing  in a ball of radius $R$. Then, $|Y| \le (\frac{R}{r})^{2\ddim(X)}$.
\end{fact}

Write $\ell = \ceil{\log \Delta}$,
and let $\{N_i\}_{i \geq 0}^{\ell}$ be a sequence
of hierarchical nets, where $N_0 = X$ and for each $i \in [\ell]$, $N_i$ is
a $2^{i}$-net for $N_{i-1}$.
We refer to $N_i$ as the \emph{$i$-level net},
and the points of $N_i$ are called the \emph{$i$-level net-points}.
Note that $N_0 = X \supseteq N_1 \supseteq \ldots \supseteq N_\ell$,
and $N_\ell$ contains exactly one point.
The same point of $X$ may have instances in many nets; specifically, an $i$-level net-point is necessarily a $j$-level net-point,
for every $j \in [0,i]$. When we wish to refer to a specific instance of a point $p \in X$, 
which is determined uniquely by some level $i \in [0,\ell]$ (such that $p \in N_i$), we may denote it by the pair $(p,i)$.  

\paragraph{Net-tree Spanner.}
For each   $i \in [0,\ell-1]$,
\emph{cross edges} are added between $i$-level net-points that are within distance $\gamma \cdot 2^i$ from each other,
for some parameter $\gamma = \Theta(\frac{1}{\eps})$.
By Fact \ref{prop:small_net}, for each $i \in [0,\ell]$, the degree of any net-point $(p,i)$ in the net-tree spanner 
(due to $i$-level cross edges) is $\eps^{-O(\ddim)}$.

Let $R$ be the sum of radii of all net-points, where 
the radius $rad(p,i)$ of an $i$-level net-point $(p,i)$ is equal to $2^i$, disregarding the single net-point at level $\ell$.
It is easy to see that the weight of the net-tree spanner is given by $R \cdot \gamma \cdot \eps^{-\Theta(\ddim)} = R \cdot 
\eps^{-\Theta(\ddim)}$. However, as shown in Section \ref{sec:int}, $R$ may be as large as $\Theta(\log n) \cdot 
\omega(MST(X,\delta))$, even for 1-dimensional Euclidean metrics.
Consequently, the net-tree spanner in doubling metrics has lightness $\Theta(\log n) \cdot \eps^{-\Theta(\ddim)}$. 
On the other hand, it turns out that in \emph{snowflake} doubling metrics -- the net-tree spanner is light.

\section{Proof via $\ell_p$ space}\label{sec:first}

In this section we present our first proof that snowflake metrics admit light
$(1+\epsilon)$-spanner. We will first need to prove a certain property 
concerning Euclidean vectors (Section \ref{sec:prop}). We then show that given a
set of $d$-dimensional vectors $S$, the edges of a light $(1+\epsilon)$-spanner for $(S,\ell_2)$ 
also form a light $(1+\epsilon^{O(1)})$-spanner for all $(S,\ell_p)$ 
$p \ge 1$, which itself implies a light spanner for snowflake metrics.

Before presenting the proof for $\ell_p$ spaces,
we formally state the Euclidean light spanner theorem of \cite{DHN93}.
For an edge set $E$, let $w_p(E)$ be the sum of the lengths of the edges under $\ell_p$.
Let $\MST_p(S)$ be the edge set of the minimum spanning tree for $S$ under $\ell_p$.

\begin{theorem}\label{thm:l2}
Let $S$ be a $d$-dimensional point set of size $n$. Then there exists an edge set
$E$ which forms a $(1+\eps)$-spanner for $S$ under $\ell_2$, with
$w_2(E) = \eps^{-O(d)} w_2(\MST_2(S))$.
\end{theorem}

\subsection{Properties of Euclidean vectors}\label{sec:prop}

In Section \ref{sec:first-main} we will require a property of Euclidean vectors
detailed in Lemma \ref{lem:ineq} below. We begin with the following simple fact:

\begin{lemma}\label{lem:pre}
For $0 \le \epsilon_0 \le \epsilon \le \epsilon_1$, $\epsilon \le \frac{1}{4}$,
and non-negative $a,b$, if
$\epsilon_1 a + \eps_0 b \le \epsilon(a+b)$
then
$\epsilon_1 a + \sqrt{\epsilon_0} b \le \sqrt{\eps}(a+b)$.
\end{lemma}

\begin{proof}
Rewriting both inequalities, we claim that
$(\eps_1 - \eps) a \le (\epsilon - \epsilon_0) b$
implies that 
$(\eps_1 - \sqrt{\epsilon}) a \le (\sqrt{\epsilon} - \sqrt{\epsilon_0}) b.$
This claim is confirmed by dividing both sides of the first inequality by 
$\sqrt{\eps} + \sqrt{\eps_0}$. The division immediately yields the 
right hand side of the second inequality; the left hand side follows
by noting that
$\sqrt{\eps} + \sqrt{\eps_0} \le \frac{1}{2} +  \frac{1}{2} = 1$, 
and then observing that
$\frac{\eps_1 - \eps}{\sqrt{\eps} + \sqrt{\eps_0}}
> \frac{\eps_1(\sqrt{\eps} + \sqrt{\eps_0}) 
	- \sqrt{\eps}(\sqrt{\eps} + \sqrt{\eps_0})}{\sqrt{\eps} + \sqrt{\eps_0}}
= \eps_1- \sqrt{\eps}$.
\QED
\end{proof}

\begin{lemma}\label{lem:ineq}
Let $V$ be a set of $d$-dimensional vectors. Given some vector $w$, let 
each vector $v_i \in V$ be decomposed into
two orthogonal vectors $v_i^\perp, v_i^\parallel$, where 
$v_i^\perp + v_i^\parallel = v$
and $v_i^\perp$ is orthogonal to $w$. 
For $0< \epsilon \le \frac{1}{4}$, if

$$\sum_{v \in V} \| v \|_2 = (1+\epsilon) \| \sum_{v \in V} v^\parallel \|_2,$$

then

$$ \sum_{v \in V} \| v^\parallel \|_2 \le (1+\epsilon) \| \sum_{v \in V} v^\parallel \|_2$$

and 

$$\sum_{v \in V} \| v^\perp \|_2 \le 3(1+\epsilon)\sqrt{\epsilon} \| \sum_{v \in V} v^\parallel \|_2.$$
\end{lemma}

\begin{proof}
The first part of the Lemma is trivial: Since the vectors are orthogonal,
$\sum_{v \in V} \| v^\parallel \|_2
\le \sum_{v \in V} \sqrt{\| v^\parallel \|_2^2 + \| v^\perp \|_2^2}
= \sum_{v \in V} \| v \|_2 
= (1+\epsilon) \| \sum_{v \in V} v^\parallel \|_2$.

Moving to the second part, note that we may assume without loss of generality that all
$v_i^\parallel$ have the same length: We can always enforce this property by 
segmenting the $v_i$'s into small vectors with equal parallel contribution 
without violating the conditions of the lemma.

Define set $A$, where for each element $a_i \in A$, 
$a_i = \frac{\| v_i^\perp \|_2}{\| v_i^\parallel \|_2}$.
We have
$\| v_i \|_2 
= \sqrt{\| v_i^\parallel \|_2^2 + \| v_i^\perp \|_2^2 }
= \sqrt{1+a_i^2} \| v_i^\parallel \|_2$.
We wish to bound 
$\sum_{v_i \in V} \| v_i^\perp \|_2 
= \sum_{v_i \in V} a_i \| v_i^\parallel \|_2$.
Partition $A$ into two subsets 
$A_0,A_1 \subset A$, where $A_0$ contains elements of value less than 1, and $A_1$ 
contains elements of value greater or equal to 1. 
Likewise, partition $V$ into two subsets $V_0,V_1 \subset V$, where
$v_i \in V_j$ if and only if $a_i \in A_j$. We have

$$\sum_{v_i \in V_1} \| v_i \|_2 
= \sum_{v_i \in V_1} \sqrt{1+a_i^2}  \| v_i^\parallel \|_2
> \sum_{v_i \in V_1} (1+\frac{a_i}{3}) \| v_i^\parallel \|_2,$$

$$\sum_{v_i \in V_0} \| v_i \|_2
= \sum_{v_i \in V_0} \sqrt{1+a_i^2} \| v^\parallel \|_2	
> \sum_{v_i \in V_0} (1+\frac{a_i^2}{3}) \| v^\parallel \|_2.$$

By the assumption of the lemma, 

\smallskip
$\begin{array}{lcl}
(1+\epsilon) \| \sum_{v \in V} v^\parallel \|_2
&=& \sum_{v_i \in V} \| v_i \|_2							\\
&=& \sum_{v_i \in V_1} \| v_i \|_2 + \sum_{v_i \in V_0} \| v_i \|_2       \\
&>& \sum_{v_i \in V_1} (1+\frac{a_i}{3}) \| v_i^\parallel \|_2
        + \sum_{v_i \in V_0} (1+\frac{a_i^2}{3}) \| v_i^\parallel \|_2              \\
&=& \sum_{v_i \in V} \| v_i^\parallel \|_2
        + \sum_{v_i \in V_1} \frac{a_i}{3} \| v_i^\parallel \|_2
        + \sum_{v_i \in V_0} \frac{a_i^2}{3} \| v_i^\parallel \|_2			\\
&\ge& \| \sum_{v_i \in V} v_i^\parallel \|_2
        + \sum_{v_i \in V_1} \frac{a_i}{3} \| v_i^\parallel \|_2
        + \sum_{v_i \in V_0} \frac{a_i^2}{3} \| v_i^\parallel \|_2.             \\
\end{array}$
\smallskip

The final inequality follows from the triangle inequality.
It follows immediately that 
$\sum_{v_i \in V_1} \frac{a_i}{3} \| v_i^\parallel \|_2
+ \sum_{v_i \in V_0} \frac{a_i^2}{3} \| v_i^\parallel \|_2
\le \epsilon \| \sum_{v \in V} v_i^\parallel \|_2.$
Recall that we may assume that all $\| v_i^\parallel \|_2$ are equal, so 
$\sum_{a_i \in A_1} \frac{a_i}{3} 
+ \sum_{a_i \in A_0} \frac{a_i^2}{3}
\le \epsilon |A| 
= \epsilon |V|$.
Set
$\sum_{ a_i \in A_1} \frac{a_i}{3} = \epsilon_1 |V_1|$
and
$\sum_{ a_i \in A_0} \frac{a_i^2}{3} = \epsilon_0 |V_0|$,
so that
$ \epsilon_1 |V_1| + \epsilon_0 |V_0| \le \epsilon |V| $
and it must be that 
$\epsilon_1 \ge \frac{1}{3} > \epsilon$ 
and 
$\epsilon_0 \le \epsilon$. 
If we fix the term $\eps_0$ and take the elements of $A_0$ as variable, we see that
$ \sum_{a_i \in A_0} a_i $
attains its maximum value when $a_i = \sqrt{3\epsilon_0}$ for all $a_i \in A_0$. So
$ \sum_{a_i \in A_0} a_i \le \sqrt{3\epsilon_0} |V_0|$, and using Lemma \ref{lem:pre} we have that 

\smallskip

$\begin{array}{lcl}
\sum_{v \in V} \| v^\perp \|_2
&=& \sum_{v_i \in V_1} a_i \| v_i^\parallel \|_2
  + \sum_{v_i \in V_0} a_i \| v_i^\parallel \|_2		\\
&\le& [3\epsilon_1 |V_1| + \sqrt{3 \eps_0} |V_0|] \| v_0^\parallel \|_2	\\
&\le& 3[\epsilon_1 |V_1|+\sqrt{\epsilon_0} |V_0|] \| v_0^\parallel \|_2	\\
&\le& 3\sqrt{\epsilon} |V| \| v_0^\parallel \|_2			\\
&\le& 3(1+\epsilon)\sqrt{\epsilon} \| \sum_{v_i \in V} v_i^\parallel \|_2.
\end{array}$

\smallskip

\QED
\end{proof}

\subsection{Light $\ell_p$ and snowflake spanners}\label{sec:first-main}

Here we will show that $\ell_p$ spaces -- and therefore snowflake
metrics -- admit light spanners.
Observe that for any $d$-dimensional vector $x$, 
when $p \ge 2$ we have 
$\| x \|_p \le \| x \|_2 \le d^{\frac{1}{2}-\frac{1}{p}} \| x \|_p$,
and when $p \le 2$ we have
$d^{\frac{1}{2} - \frac{1}{p}} \| x \|_p \le \| x \|_2 \le \| x \|_p$.
We can prove the following lemma:

\begin{lemma}\label{lem:lp}
Let $E$ be the edge set of a $(1+\eps)$-spanner for $S$ under $\ell_2$, with weight
$w_2(E) = c \cdot w_2(\MST_2(S))$ for some constant $c := c(d,\eps)$. 
Let $d' = \max \{ d^{\frac{1}{2} - \frac{1}{p}}, d^{\frac{1}{p} - \frac{1}{2}}\}$.
Then set $E$ forms a $(1+\epsilon) (1+3d'\sqrt{\eps})$-spanner for $S$ under $\ell_p$, 
$1 \le p \le \infty$ with weight 
$w_p(E) \le c d' \cdot w_p(\MST_p(S))$. 
\end{lemma}

\begin{proof}
We first prove the weight guarantee. Recall that $\MST_2(S)$ is the minimum weight
connected graph under $\ell_2$. When $p \ge 2$ we have

\smallskip

$
\begin{array}{lcl}
w_p (E) 
& \le & w_2(E) 				\\
& = &	c \cdot w_2(\MST_2(S)) 		\\
& \le & c \cdot w_2(\MST_p(S))		\\
& \le & c d^{\frac{1}{2} - \frac{1}{p}} \cdot w_p(\MST_p(S)).
\end{array}
$

\smallskip

When $p \le 2$, we have 

\smallskip

$
\begin{array}{lcl}
w_p (E) 
& \le & d^{\frac{1}{p} - \frac{1}{2}} \cdot w_2(E) 			\\
& = & 	c d^{\frac{1}{p} - \frac{1}{2}} \cdot w_2(\MST_2(S)) 	\\
& \le &	c d^{\frac{1}{p} - \frac{1}{2}} \cdot w_2(\MST_p(S))		\\
& \le &	c d^{\frac{1}{p} - \frac{1}{2}} \cdot w_p(\MST_p(S)).
\end{array}
$

\bigskip

This completes the proof of lightness, and we proceed with the stretch guarantee. 
Consider any vertex pair $x_0,x_t \in S$, connected 
in $E$ by the minimal (under $\ell_2$) edge path 
$P_{0,t} = \{ (x_0,x_1),\ldots,(x_{t-1},x_t) \}$. 
Let $V$ be a set of vectors $v_i$ transitioning $x_i$ to $x_{i+1}$, that
is $v_i = x_{i+1} - x_i$. We will employ Lemma \ref{lem:ineq} with respect to $V$ and 
vector $w = x_t - x_0$. For parallel vectors, we have that 
$\frac{\sum_{v \in V} \| v^\parallel \|_2}{\sum_{v \in V} \| v^\parallel \|_p}
= \frac{\| \sum_{v \in V} v^\parallel \|_2}{\| \sum_{v \in V} v^\parallel \|_p}
$.
Since $E$ is a $(1+\eps)$-spanner, Lemma \ref{lem:ineq} gives 
$ \sum_{v \in V} \| v^\parallel \|_2 
\le (1+\epsilon) \| \sum_{v \in V} v^\parallel \|_2$
and so
$ \sum_{v \in V} \| v^\parallel \|_p
\le (1+\epsilon) \| \sum_{v \in V} v^\parallel \|_p$.
Lemma \ref{lem:ineq} also gives
$\sum_{v \in V} \| v^\perp \|_2 \le 3(1+\epsilon)\sqrt{\epsilon} \| \sum_{v \in V} v^\parallel \|_2$:
When $p \ge 2$, we have
$\sum_{v \in V} \| v^\perp \|_p 
\le \sum_{v \in V} \| v^\perp \|_2 
\le 3(1+\epsilon)\sqrt{\epsilon} \| \sum_{v \in V} v^\parallel \|_2
\le 3d'(1+\epsilon)\sqrt{\epsilon} \| \sum_{v \in V} v^\parallel \|_p$, 
and when $p < 2$ we have
$\sum_{v \in V} \| v^\perp \|_p 
\le d' \sum_{v \in V} \| v^\perp \|_2 
\le 3d'(1+\epsilon)\sqrt{\epsilon} \| \sum_{v \in V} v^\parallel \|_2
\le 3d'(1+\epsilon)\sqrt{\epsilon} \| \sum_{v \in V} v^\parallel \|_p$.
We conclude that

\smallskip

$
\begin{array}{lcl}
w_p(E) 
&=& \sum_{i=0}^{t-1} \| v_{i} \|_p		\\
&\le& \sum_{i=0}^{t-1} [ \| v^\parallel_{i} \|_p + \| v^\perp_{i} \|_p ]	\\
&\le& (1+\eps) \| w \|_p + (3 d' (1+\epsilon) \sqrt{\eps}) \| w \|_p		\\
&=& (1+\epsilon) (1+3d'\sqrt{\eps}) \| w \|_p. 
\end{array}
$
\\
\QED
\end{proof}

It follows that a $(1+\eps)$-spanner for $(S,\ell_p)$ with lightness $\epsilon^{-\tilde{O}(d)}$
can be achieved by building a $(1+O((\eps/d')^2))$-spanner for $(S,\ell_2)$ of lightness
$\epsilon^{-O(d)}$, as in Theorem \ref{thm:l2}.

\paragraph{Remark.} 
Lemma \ref{lem:lp} shows that the stretch guarantee of a Euclidean spanner holds
even if the metric is later changed to a different $\ell_p$. While this claim is true
for Euclidean spanners, it does not hold for all $\ell_p$ spanners, for example 
$\ell_\infty$: Consider three vectors 
$v_1 = (0,0)$, $v_2 = (1,1)$, $v_3 = (2,0)$, and the two edges
$\{v_1,v_2\}$ and $\{v_2,v_3\}$. Then this spanner has no distortion
for $(S,\ell_\infty)$, but constant distortion for $(S,\ell_p)$ and fixed
$p$.

Given a metric $M=(X,\delta)$ with dimension $\ddim(M)$, the metric 
$(X,\delta^\alpha)$ has doubling dimension $\frac{\ddim(M)}{\alpha}$ \cite{GK-11}.
Har-Peled and Mendel \cite{HM-06} demonstrated that $\frac{1}{2}$-snowflake metrics embed into 
$\ell_\infty$ with low distortion and dimension, and their result can be extended to show that
$(X,\delta^\alpha)$ embeds into $d$-dimensional $\ell_\infty$ with distortion $1+\eps$ and target dimension 
$d = \frac{\eps^{-O(\ddim(M)/\alpha)}}{1-\alpha}$ \cite{GK-11}. 
Together with Theorem \ref{thm:l2} and Lemma \ref{lem:lp}, we may conclude:

\begin{theorem}\label{thm:first-main}
Let $M=(X,\delta)$ be an $n$-point doubling metric with doubling dimension $d$.
For any $0 < \alpha < 1$ and $\eps > 0$,
there exists a $(1+\eps)$-spanner for $(X,\delta^\alpha)$ with lightness 
$2^{\frac{1}{1-\alpha} \eps^{-\tilde{O}(\ddim(M)/\alpha)}}$.
\end{theorem}

\section{Direct snowflake proof}\label{sec:second}

In this section we present a second, tighter proof for the existence of light spanners for
snowflake metrics. Let $(X,\delta)$ be an arbitrary $n$-point doubling metric,
and let $0 < \alpha < 1$ be an arbitrary parameter.
The corresponding snowflake metric $\tilde X := (X,\delta^\alpha)$ is also doubling,
with doubling dimension at most $\frac{\ddim(X)}{\alpha}$.
Denote by $\tilde \Delta = \max_{u, v \in X}\delta^\alpha(u, v)$ the \emph{diameter} of $\tilde X$,
and write $\tilde \ell = \ceil{\log \tilde \Delta}$.
Next we build the sequence of hierarchical nets $\{\tilde N_i\}_{i \geq 0}^{\tilde \ell}$ 
and the net-tree spanner for the snowflake metric $\tilde X$;
denote by $\tilde R$ the sum of radii of all net-points, disregarding the single net-point at level $\tilde \ell$.

In what follows we prove that $\tilde R = O(\frac{1}{1 - \alpha}) \cdot \omega(MST(\tilde X)$, which implies that the lightness 
of the net-tree spanner for $\tilde X$ is $\eps^{-O(\frac{\ddim(X)}{\alpha})} \cdot \frac{1}{1-\alpha}$.
We first give some intuition by considering a trivial metric in Section \ref{sec:int}. The proof of the general case 
proceeds in two stages.
\vspace{-0.1in}
\begin{enumerate}[leftmargin=*]\itemsep2pt \parskip1pt \parsep1pt
\item In the first stage (Section \ref{sec:stage1}) 
we construct an auxiliary graph $\tilde G$ whose weight 
$\tilde W$ is $\Omega(\tilde R)$.
The graph $\tilde G$ will be given as the union of $\tilde\ell$ simple paths,
and is thus more amenable to analysis than the standard net-tree spanner.
The vertex set of this graph $\tilde G$ is equal to $X$.
\item In the second stage (Section \ref{sec:stage2}) we show that $\tilde W = O(\frac{1}{1-\alpha}) \cdot \omega(MST(\tilde X)$.
\end{enumerate}

\subsection{Intuition and High-Level Ideas} \label{sec:int}
Before delving into the proof,  it is instructive to consider
the 1-dimensional Euclidean case, and to restrict our attention to $\alpha = \frac{1}{2}$. 
The intuition behind the general proof is present even in this basic case.

Let $\vartheta = \vartheta_n$ be a set of $n$ points $v_1,\ldots,v_n$ lying on the $x$-axis with coordinates $1,\ldots,n$,
respectively, 
and consider the corresponding
snowflake metric $\tilde \vartheta = (\vartheta, \ell_2^{\frac{1}{2}})$.
(See Figure \ref{vartheta} for an illustration.)
\begin{figure*}[htp]
\begin{center}
\begin{minipage}{\textwidth} 
\begin{center}
\setlength{\epsfxsize}{4.1in} \epsfbox{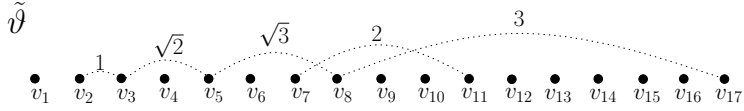}
\end{center}
\end{minipage}
\caption[]{ \label{vartheta} \sf \footnotesize An illustration of the 1-dimensional Euclidean point set 
$\vartheta = \vartheta_{17} = \{v_1,\ldots,v_{17}\}$, and the corresponding snowflake metric 
$\tilde \vartheta = (\vartheta, \ell_2^{\frac{1}{2}})$.
Only 6 edges (out of the total ${17 \choose 2}$ edges) are depicted in the figure, along with their weights
which are set according to the snowflake distance function $\|\cdot\|^{\frac{1}{2}}$;
thus the weight of edge $(v_1,v_{17})$, for example, is given by $|17-1|^{\frac{1}{2}} = 4$.}
\end{center}
\vspace{-0.24in}
\end{figure*}
Since the diameter $\tilde \Delta$ of $\tilde \vartheta$ is $\sqrt{n-1}$,  
the number of levels in the underlying net-tree is given by $\tilde \ell = \lceil \log \tilde \Delta \rceil = \lceil \frac{1}{2} \log 
(n-1) \rceil$. Also, the doubling property implies that the number of $i$-level net-points is proportional to $\frac{n}{2^{2i}}$.
If $\tilde R_i$ stands for the sum of radii of all $i$-level net-points, then  $\tilde R_i = O(\frac{n}{2^{2i}}) \cdot 2^i = O(\frac{n}{2^i})$, so $\tilde R = \sum_{i = 0}^{\tilde \ell-1} \tilde R_i = \sum_{i = 0}^{\tilde \ell-1} O(\frac{n}{2^i}) = O(n)$.
Observe that $\omega(MST(\tilde \vartheta)) = n-1$.
It follows that $\tilde R = O(\omega(MST(\tilde \vartheta)))$, which proves the desired bound on the lightness.\footnote{This argument does not work for the original (non-snowflake) metric $\vartheta = (\vartheta, \|\cdot\|)$,
as there the number of $i$-level net-points is proportional to $\frac{n}{2^{i}}$. Thus the sum of radii of all $i$-level net-points is constant, and therefore $R = \Omega(\log n)$.}

This argument is simple, but
it is unclear how to generalize it for arbitrary snowflake metrics.
We next provide a more involved proof, whose intuition will be used in the general proof.

Suppose for simplicity of the presentation that $n-1$ is an integer power of $4$.
For each index $i \in [0,\tilde \ell-1]$, we choose a subset 
$P_i = \{v_1, v_{1+ 2^{2i}}, v_{1 + 2 \cdot 2^{2i}}, \ldots,v_n\}$
of \emph{pivots} from $\tilde \vartheta$, where the $\ell_2$ distance between any two consecutive $i$-level pivots is exactly $2^{2i}$.
Since $\tilde N_i$ is a $2^i$-packing, the $\ell_2$ distance between any pair of $i$-level net-points is greater than
$2^{2i}$. Consequently, 
it is easy to see that the number $|P_i|$ of $i$-level pivots is greater than 
the number $|\tilde N_i|$ of $i$-level and net-points, i.e., $|P_i| > |\tilde N_i|$.
Let $\Pi_i$ be the simple path connecting all points of $P_i$, i.e., $\Pi_i = (v_{1},v_{1+2^{2i}},v_{1+2\cdot2^{2i}},\ldots,v_n)$.
Notice that the weight of each edge in $\Pi_i$ (under the snowflake distance function) is equal to the radius $2^i$
of $i$-level net-points.
Since $|P_i| \ge |\tilde N_i|$, 
the weight $\omega(\Pi_i)$ of $\Pi_i$ is no smaller than $\tilde R_i$.
Let $\tilde G = \bigcup_{i =0}^{\tilde \ell-1} \Pi_i$ be the union of the $\tilde \ell$ paths $\Pi_0,\ldots,\Pi_{\tilde \ell-1}$,
and let $\tilde W = \sum_{i =0}^{\tilde \ell-1} \omega(\Pi_i)$ be the weight of $\tilde G$.
Observe that $\tilde W \ge \tilde R$.
(See Figure \ref{tildeG}.) 
\begin{figure*}[htp]
\begin{center}
\begin{minipage}{\textwidth} 
\begin{center}
\setlength{\epsfxsize}{4.1in} \epsfbox{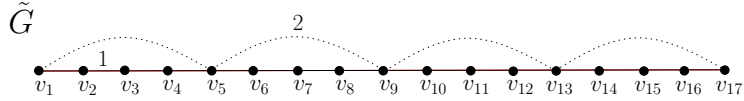}
\end{center}
\end{minipage}
\caption[]{ \label{tildeG} \sf \footnotesize An illustration of $\tilde G 
= \Pi_0 \cup \Pi_1$, in the case that $\vartheta = \{v_1,\ldots,v_{17}\}$ and $\tilde \ell = 2$.
The path $\Pi_0 = (v_1,v_2,\ldots,v_{17})$ consists of 16 edges of unit weight depicted by solid lines,
and the path $\Pi_1 = (v_1,v_5,v_9,v_{13},v_{17})$ consists of four edges of weight 2 depicted by dotted lines.
The   weight $\tilde W$ of $\tilde G$ is equal to $\omega(\Pi_0) + \omega(\Pi_1) = 16 + 8 = 24$.}
\end{center}
\end{figure*}
\vspace{-0.1in}
We have thus reduced the problem of lower bounding $\tilde R$ to that of lower bounding $\tilde W$.

Even though lower bounding $\tilde W$ for this basic 1-dimensional case can be done by a direct calculation,
our goal here is to present a method that can be applied in the general case.

Any edge $(v_l,v_{l+1})$ of path $\Pi_0 = (v_1,\ldots,v_n)$ will be called a \emph{path edge}.
We say that   edge $e  = (v_j,v_k)$ (with $1 \le j < k \le n$) \emph{loads} a path edge $(v_l,v_{l+1})$ if $j \le l < l+1 \le k$.
Thus edge $(v_j,v_k)$ loads the $k-j$ path edges $(v_j,v_{j+1}),\ldots,(v_{k-1},v_k)$.
Next we would like to distribute the weight $\|v_j,v_k\|_2^\frac{1}{2} = |k-j|^\frac{1}{2}$ of edge $e = (v_j,v_k)$ 
over all the path edges which it loads. Specifically, the \emph{load} $\xi_{(v_l,v_{l+1})}(e)$ on path edge $(v_l,v_{l+1})$ caused 
by edge $e = (v_j,v_k) \in \tilde G$, for $j \le l < l+1 \le k$, is 
defined as 
\begin{equation} \label{load}
\xi_{(v_l,v_{l+1})}(e) ~=~ \frac{\|v_j,v_k\|_2^\frac{1}{2}}{\|v_j,v_k\|_2} ~=~ \frac{1}{|k-j|^\frac{1}{2}}.
\end{equation}

This means that if we sum up the loads of the $k-j$ path edges $(v_j,v_{j+1}),\ldots,(v_{k-1},v_k)$ due to edge $e = (v_j,v_k)$,
we have the weight $\|v_j,v_k\|_2^\frac{1}{2} = |k-j|^\frac{1}{2}$ of that edge. 
For each path edge $(v_l,v_{l+1})$, the sum of loads on that edge caused by all edges $e \in \tilde G$
is called the \emph{load of $(v_l,v_{l+1})$ by $\tilde G$}, and is denoted by $\xi_{(v_l,v_{l+1})} = \xi_{(v_l,v_{l+1})}(\tilde G)$. 
The \emph{load of the graph $\tilde G$}, $\xi(\tilde G)$, is the sum of loads over all path edges $(v_l,v_{l+1})$ 
by $\tilde G$, i.e., $\xi(\tilde G) ~=~ \sum_{l \in [n-1]} \xi_{(v_l,v_{l+1})}(\tilde G).$ 
A double counting argument  (see Observation \ref{ob2} below for more details) yields $\tilde W = \xi(\tilde G)$.
We have thus reduced the problem of lower bounding $\tilde R$ to that of lower bounding the load $\xi(\tilde G)$
of $\tilde G$.

To get some intuition, consider first the graph $\tilde G$ from Figure \ref{tildeG}, which corresponds to the case $n = 17$.
It is easy to see that each path edge is loaded by a single edge of $\Pi_i$, for each $i \in [0,1]$.
For example, edge $(v_1,v_2)$ is loaded by edge $(v_1,v_2)$ of path $\Pi_0$ and edge $(v_1,v_5)$ of path $\Pi_1$,
and we thus have 
$$\xi_{(v_1,v_2)}(\tilde G) ~=~ \xi_{(v_1,v_2)}(v_1,v_2) + \xi_{(v_1,v_2)}(v_1,v_5)  
~=~ 1 + \frac{1}{2} ~=~ \frac{3}{2}.$$
Thus $\tilde W = \xi(\tilde G) = \sum_{l \in [16]} \xi_{(v_l,v_{l+1})}(\tilde G) = \frac{3}{2} \cdot 16 = \frac{3}{2} \cdot  \omega(MST(\tilde \vartheta)).$
We turn to the case of general $n$. 
  
Consider any path edge $(v_l,v_{l+1})$. It is loaded by a single edge $e_i$ of $\Pi_i$, for each level $i \in [0,\tilde \ell-1]$,
and its load caused by edge $e_i$ is given by $\xi_{(v_l,v_{l+1})}(e_i)  ~=~ \frac{1}{2^i}$. 
Summing over all $\tilde \ell$ levels, we have $$\xi_{(v_l,v_{l+1})}(\tilde G) ~=~ \sum_{i = 0}^{\tilde \ell-1} \xi_{(v_l,v_{l+1})}(e_i) 
~=~ \sum_{i = 0}^{\tilde \ell-1} \frac{1}{2^i} ~\le~ 2.$$
We conclude that $$\tilde W ~=~ \xi(\tilde G) ~=~ \sum_{l \in [n-1]} \xi_{(v_l,v_{l+1})}(\tilde G) ~\le~ 2(n-1) ~=~ 2 \cdot \omega(MST(\tilde \vartheta)).$$


\subsection{Stage I}\label{sec:stage1}

We will now analyze an arbitrary snowflake metric $\tilde X = (X,\delta^\alpha$), where $0 < \alpha < 1$.
In this first stage we will construct an auxiliary graph $\tilde G = (X,\tilde E,\tilde w)$ 
whose weight $\tilde W = \tilde \omega(\tilde G)$ is at least as large as $\tilde R$ (up to a constant).
The graph $\tilde G$ will be given as the union of $\tilde \ell$ simple paths,
and is thus more convenient for analysis purposes than the standard net-tree spanner.
The vertex set of this graph $\tilde G$ is equal to $X$, and 
$\tilde G$ is equipped with a weight function $\tilde w$ which is dominated by the distance function $\delta^\alpha$.  

To mimic the 1-dimensional case, 
we start by computing a Hamiltonian path $\Pi = (v_1,\ldots,v_n)$ for $\tilde X$ of weight $\omega(\Pi) =  \sum_{i \in [n-1]} \delta^\alpha(v_i,v_{i+1})$ at most $O(\omega(MST(\tilde X)))$. (Computing such a path $\Pi$ 
is a standard procedure, which can be easily carried out given a constant-factor approximation MST for $\tilde X$.) 

As before, we'd like to compute a set $P_i$ of $i$-level pivots, for all $i \in [0,\tilde \ell-1]$ -- but 
this process will be carried out more carefully now.
Having done that, the graph $\tilde G$ will be obtained as before, from the union of the $\tilde \ell$ paths 
$\Pi_0,\ldots,\Pi_{\tilde \ell-1}$, with each $\Pi_i$ connecting the set $P_i$ of $i$-level pivots via  a simple path.

We show how to compute the set $P_i$ of $i$-level pivots, for $i \in [0,\tilde \ell-1]$.
Recall that in the 1-dimensional   case, any two consecutive $i$-level pivots are at (snowflake) 
distance \emph{exactly} $2^{i}$ apart; we cannot achieve this property in the general case.
Instead, we will make sure that the distance between consecutive pivots will be at least $2^{i-1}$.
(The reason we use a distance threshold of $2^{i-1}$ rather than $2^i$ is technical -- this 
enables us to guarantee that $|P_i| \ge |\tilde N_i|$.)

For $i = 0$, we simply take $P_0 = X = \{v_1,\ldots,v_n\}$, and $\Pi_0 = \Pi = (v_1,\ldots,v_n)$.
Next consider $i \in [\tilde \ell-1]$.
The first $i$-level pivot $p^{(i)}_1$ is $v_1$. Having assigned the $j$ first $i$-level pivots 
$p^{(i)}_1,\ldots,p^{(i)}_j$,  the next pivot
$p^{(i)}_{j+1}$ is the first point after $p^{(i)}_{j}$ in $\Pi$ which is at distance at least $2^{i-1}$ from it.
Formally, let $k$ be the index such that $p^{(i)}_j = v_k$,
and let $k'$ be the smallest index after $k$ such that $\delta^\alpha(p^{(i)}_j,v_{k'}) \ge 2^{i-1}$. 
Then $p^{(i)}_{j+1} = v_{k'}$

For $i \in [0,\tilde \ell-1]$, the $i$-level path $\Pi_i$
is a simple path over the $i$-level pivots.
Although an edge $e$ may be very long, 
we set the weight $\tilde \omega(e)$ of each edge $e$ of path $\Pi_i$ to be $2^{i-1}$, 
and it follows that the edge weights
$\tilde \omega(e)$
of edges $e \in \Pi_i$ will be dominated by the corresponding snowflake distances.
(Indeed, by construction, the snowflake distance between any pair of consecutive $i$-level pivots is at least $2^{i-1}$.)
Recall that in the 1-dimensional case we used the weights $\omega(\cdot)$ as given by the snowflake distances -- here we use different weights $\tilde \omega(\cdot)$ that are dominated
by $\omega(\cdot)$.
Denote the edge set of $\Pi_i$ by $\tilde E_i$, with $|\tilde E_i| = |P_i| - 1$.
Let $\tilde G = (X,\tilde E, \tilde \omega)$ be the graph obtained from the union of the $\tilde \ell$ paths
$\Pi_0,\ldots,\Pi_{\tilde \ell-1}$, with $\tilde E = \bigcup_{i = 0}^{\tilde \ell-1} \tilde E_i$,
and $\tilde \omega$ is the weight function as defined above.



As in the 1-dimensional case, we next show that the weight $\tilde W$ of $\tilde G$ is not much smaller than $\tilde R$. 

The analysis starts with the next observation, which follows immediately from the construction.
\begin{observation} \label{ob}
Let $p^{(i)}_j = v_k$ and $p^{(i)}_{j+1} = v_{k'}$ be arbitrary consecutive $i$-level pivots,
with $k < k'$.
Then for any index $j \in [k,k' -1]$, $\delta^\alpha(p^{(i)}_{j},v_{j}) < 2^{i-1}$.
Hence for any $j,\hat j \in [k,k' -1]$, $\delta^\alpha(v_{j},v_{\hat j}) < 2^{i}$.
\end{observation}

We argue that the number of $i$-level pivots is no smaller than the number of $i$-level net-points.
\begin{lemma} \label{lm0}
$|P_i| \ge |\tilde N_i|$. 
\end{lemma}
\begin{proof}
Since the $i$-level net $\tilde N_i$ is a $2^i$-packing, any two $i$-level net-points are at distance at least $2^i$ apart.
By Observation \ref{ob}, for any two consecutive $i$-level pivots
$p^{(i)}_j = v_k$ and $p^{(i)}_{j+1} = v_{k'}$, with $k < k'$, at most one point from $\{v_k,\ldots,v_{k'-1}\}$
belongs to $\tilde N_i$. The lemma follows.
\QED
\end{proof}

\begin{lemma} \label{lm2}
For each $i \in [0,\tilde \ell-1]$, $|\tilde E_i| ~=~ |P_i|-1 ~\ge~ \frac{1}{2} \cdot |P_i| ~\ge~ \frac{1}{2} \cdot |\tilde N_i|$.
\end{lemma}
\begin{proof}
The first equality is immediate from the construction and the third inequality follows from Lemma \ref{lm0}. 
In what follows we prove the second inequality $|P_i|-1 ~\ge~ \frac{1}{2} \cdot |P_i|$, or equivalently $|P_i| \ge 2$.


If all indices $k \in [2,n]$ satisfied $\delta^\alpha(p^{(i)}_{1} = v_1,v_k) < 2^{i-1}$,
the distance between any two points of $X$ would be smaller than $2^i < \tilde \Delta$, a contradiction.
Let $k'$ be the smallest index for which $\delta^\alpha(p^{(i)}_{1} = v_1,v_{k'}) \ge 2^{i-1}$.
By construction, $v_{k'}$ will be the second $i$-level pivot $p^{(i)}_{2}$.
Hence $|P_i| \ge 2$, and we are done.
\QED
\end{proof}

Lemma  \ref{lm2} implies that the weight $\tilde W = \tilde \omega(\tilde G)$ of $\tilde G$ is not much smaller than $\tilde R$.
\begin{corollary} \label{corwt}
$\tilde W \ge \frac{\tilde R}{4}$.
\end{corollary}
\begin{proof}
For each $i \in [0,\tilde \ell-1]$, we denote by $\tilde R_i$ the sum of radii of all $i$-level net-points. Observe that $\tilde R_i = |\tilde N_i| \cdot 2^i$.
By Lemma \ref{lm2} and the construction, we have
$$\tilde  W ~=~ \sum_{i \in [0,\tilde \ell-1]} \tilde \omega(\Pi_i) ~=~  \sum_{i \in [0,\tilde \ell-1]} |\tilde E_i| \cdot 2^{i-1}
~\ge~ \sum_{i \in [0,\tilde \ell-1]} \frac{1}{2} \cdot |\tilde N_i| \cdot 2^{i-1} ~=~ \sum_{i \in [0,\tilde \ell-1]}\frac{1}{4} \cdot \tilde R_i ~=~ \frac{\tilde R}{4}. \inQED$$
\end{proof}

\subsection{Stage II} \label{sec:stage2}

In this second stage we show that $\tilde W = O(\frac{1}{1-\alpha}) \cdot \omega(MST(\tilde X)$.
We use a charging scheme, which generalizes the one used in Section \ref{sec:int}:
\begin{enumerate}
\item First, we distribute the weight of each edge of $\tilde G$ between the path edges in $\Pi$ that it ``loads''.
(This is where we leverage on the fact that the metric $\tilde X  = (X,\delta^\alpha)$ is a snowflake, by using the original distance function $\delta$.)
\item Second, we show that the load incurred in this way by each path edge $(v_l,v_{l+1})$ is  $O(\frac{1}{1-\alpha}) \cdot \delta^\alpha(v_l,v_{l+1})$.
This implies that the weight $\tilde W = \tilde \omega(\tilde G)$ of $\tilde G$ does not exceed the weight $\omega(\Pi)$ of the underlying path $\Pi$ by more than 
a factor of $\frac{1}{1-\alpha}$,
thereby giving $\tilde W = O(\frac{1}{1-\alpha}) \cdot \omega(\Pi) =  O(\frac{1}{1-\alpha}) \cdot \omega(MST(\tilde X))$.
\end{enumerate}

The \emph{path distance} $\delta_\Pi(v_j,v_k)$ between a pair $v_j,v_k \in X$ of points 
is given by $\sum_{i=j}^{k-1}\delta(v_i,v_{i+1})$, i.e., the path distance is defined with respect to the original (non-snowflake) distance function $\delta$.
As in Section \ref{sec:int}, any edge $(v_l,v_{l+1})$ of path $\Pi = (v_1,\ldots,v_n)$ will be called a \emph{path edge}.
Also, we say that edge $e  = (v_j,v_k)$ (with $1 \le j < k \le n$) \emph{loads} a path edge $(v_l,v_{l+1})$ if $j \le l < l+1 \le k$. Thus edge $(v_j,v_k)$ loads the $k-j$ path edges $(v_j,v_{j+1}),\ldots,(v_{k-1},v_k)$.
Next we would like to distribute the weight $\tilde \omega(e)$ of edge $e = (v_j,v_k)$ to all the path edges that it loads. Specifically, the \emph{load} $\xi_{(v_l,v_{l+1})}(e)$ on path edge $(v_l,v_{l+1})$ caused by edge $e = (v_j,v_k) \in \tilde G$, for $j \le l < l+1 \le k$, is 
defined as 
\begin{equation} \label{load2}
\xi_{(v_l,v_{l+1})}(e) ~=~ \tilde \omega(e) \cdot  \frac{\delta(v_l,v_{l+1})}{\delta_\Pi(v_j,v_k)}.
\end{equation}
(Note that this definition generalizes  Equation (\ref{load}) from Section \ref{sec:int}.) 
It is easy to see that the weight $\tilde \omega(e)$ of edge $e$ in $\tilde G$ is distributed between all the path edges that it loads, so that the sum of loads on these path edges caused by edge $e$ is equal to $\tilde \omega(e)$. Also, the load on a specific path edge $(v_l,v_{l+1})$ caused by edge $e \in \tilde G$ is relative
to the ratio between the weight of this path edge (with respect to the original metric $\delta$) and the total weight of all the path edges that are loaded by edge $e$ 
(also with respect to the original metric $\delta$), where the latter term is exactly the path distance between the two endpoints of $e$.

For each path edge $(v_l,v_{l+1})$, the sum of loads on that edge caused by all edges $e \in \tilde G$,
is called the \emph{load of $(v_l,v_{l+1})$ by $\tilde G$},
and denoted by $\xi_{(v_l,v_{l+1})} = \xi_{(v_l,v_{l+1})}(\tilde G)$.
The \emph{load of the graph $\tilde G$}, $\xi(\tilde G)$, is the sum of loads over all path edges $(v_l,v_{l+1})$ 
by $\tilde G$, i.e., $\xi(\tilde G) ~=~ \sum_{l \in [n-1]} \xi_{(v_l,v_{l+1})}(\tilde G).$ 
A double counting yields:
\vspace{-0.14in}
\begin{observation} \label{ob2}
$\tilde W ~=~ \sum_{e \in \tilde G} \tilde \omega(e) ~=~ 
\sum_{e \in \tilde G} \sum_{l \in [n-1]} \xi_{(v_l,v_{l+1})}(e) ~=~
\sum_{l \in [n-1]} \xi_{(v_l,v_{l+1})}(\tilde G) ~=~ \xi(\tilde G)$.
\end{observation}
We have thus reduced the problem of lower bounding $\tilde R$ to that of lower bounding the load $\xi(\tilde G)$
of $\tilde G$.


We use the following lemma to complete the argument.
\begin{lemma} \label{lm3}
For any index $l \in [n-1]$, $\xi_{(v_l,v_{l+1})}(\tilde G) = O(\frac{1}{1-\alpha}) \cdot \delta^\alpha(v_l,v_{l+1})$.
\end{lemma}
\begin{proof}
Fix any index $l \in [n-1]$. 
Note that edge $(v_l,v_{l+1})$ is loaded by a single edge of $\Pi_i$, for each $i \in [0,\tilde \ell-1]$, denoted $e_i$.
Specifically, edge $e_i$ connects a pair of consecutive $i$-level pivots $p^{(i)}_{j} = v_k,p^{(i)}_{j+1} = v_{k'}$, such
that $k \le l < l+1 \le k'$, and $\delta^\alpha(p^{(i)}_{j},p^{(i)}_{j+1}) \ge 2^{i-1}$.
Note that all edges of $\Pi_i$ in $\tilde G$ have weight $2^{i-1}$, and so $\tilde \omega(e_i) = 2^{i-1}$.
Thus the load on edge $(v_l,v_{l+1})$ incurred in level $i \in [0,\tilde \ell-1]$ is given by
$$\xi_{(v_l,v_{l+1})}(e_i) ~=~ \tilde \omega(e_i) \cdot \frac{\delta(v_l,v_{l+1})}{\delta_\Pi(p^{(i)}_{j},p^{(i)}_{j+1})}
~=~ 2^{i-1} \cdot \frac{\delta(v_l,v_{l+1})}{\delta_\Pi(p^{(i)}_{j},p^{(i)}_{j+1})},$$
and so the total load on edge $(v_l,v_{l+1})$ by $\tilde G$ is equal to $$\xi_{(v_l,v_{l+1})}(\tilde G) ~=~ \sum_{i \in [0,\tilde \ell-1]} \xi_{(v_l,v_{l+1})}(e_i)
~=~  \sum_{i \in [0,\tilde \ell-1]} 2^{i-1} \cdot \frac{\delta(v_l,v_{l+1})}{\delta_\Pi(p^{(i)}_{j},p^{(i)}_{j+1})}.$$ 
Define $\eta = \delta^\alpha(v_l,v_{l+1}), t = \lceil \log \eta \rceil$. 
We first bound the load on edge $(v_l,v_{l+1})$ incurred in levels $i \in [0,t]$. 
Since $\delta(v_l,v_{l+1}) \le \delta_\Pi(p^{(i)}_{j},p^{(i)}_{j+1})$, we have $\xi_{(v_l,v_{l+1})}(e_i) \le 2^{i-1}$.
It follows that
$$\sum_{i \in [0,t]} \xi_{(v_l,v_{l+1})}(e_i) ~\le~ \sum_{i \in [0,t]} 2^{i-1} ~<~  2\eta ~=~ O(\delta^\alpha(v_l,v_{l+1})).$$
Next, we bound the load on edge $(v_l,v_{l+1})$ incurred in levels $i \in [t+1,\tilde \ell-1]$.
We have $$\delta^\alpha(p^{(i)}_{j},p^{(i)}_{j+1}) ~\ge~ 2^{i-1} =2^t \cdot 2^{i-1-t} ~\ge~ \eta \cdot 2^{i-1-t},$$
and so  $\delta(p^{(i)}_{j},p^{(i)}_{j+1}) \ge \eta^{\frac{1}{\alpha}} \cdot 2^{\frac{i-1-t}{\alpha}}$.
By the triangle inequality,  $\delta_\Pi(p^{(i)}_{j},p^{(i)}_{j+1}) \ge \delta(p^{(i)}_{j},p^{(i)}_{j+1}) \ge \eta^{\frac{1}{\alpha}} \cdot 2^{\frac{i-1-t}{\alpha}}$.
Note also that $2^{i-1} = 2^{t} \cdot 2^{i-1-t} \le 2\eta \cdot 2^{i-1-t}$.
Hence the load on edge $(v_l,v_{l+1})$ incurred in level $i \in [t+1,\tilde \ell-1]$ is given by
$$\xi_{(v_l,v_{l+1})}(e_i) ~=~ 2^{i-1}\cdot \frac{\delta(v_l,v_{l+1})}{\delta_\Pi(p^{(i)}_{j},p^{(i)}_{j+1})} 
~\le~  2\eta \cdot 2^{i-1-t} \cdot   \frac{\eta^{\frac{1}{\alpha}}}{\eta^{\frac{1}{\alpha}} \cdot 2^{\frac{i-1-t}{\alpha}}} ~=~ 2\eta \cdot 2^{(i-1-t)(1-\frac{1}{\alpha})}.$$
It follows that $$\sum_{i \in [t+1,\tilde \ell-1]} \xi_{(v_l,v_{l+1})}(e_i) ~\le~ 2\eta \cdot \sum_{i \in [t+1,\tilde \ell-1]} 2^{(i-1-t)(1-\frac{1}{\alpha})} 
~\le~  \frac{2\eta}{1-2^{(1-\frac{1}{\alpha})}}
~=~  O\left(\frac{1}{1-\alpha}\right) \cdot \delta^\alpha(v_l,v_{l+1}).$$
Summarizing, we have $$\xi_{(v_l,v_{l+1})}(\tilde G) ~=~ \sum_{i \in [0,t]} \xi_{(v_l,v_{l+1})}(e_i) + \sum_{i \in [t+1,\tilde \ell-1]} \xi_{(v_l,v_{l+1})}(e_i) ~\le~ 
O\left(\frac{1}{1-\alpha}\right) \cdot \delta^\alpha(v_l,v_{l+1}).  \inQED$$
\end{proof}
\noindent
{\bf Wrapping Up.~}
Recall that $\Pi = (v_1,\ldots,v_n)$ is a Hamiltonian path for $\tilde X$ of weight 
$\omega(\Pi) = \sum_{l \in [n-1]} \delta^\alpha(v_l,v_{l+1})$ at most $O(\omega(MST(\tilde X)))$.
Observation \ref{ob2} and Lemma \ref{lm3}  imply that $$\tilde W ~=~ \sum_{e \in \tilde G} \tilde \omega(e) ~=~ \sum_{l \in [n-1]} \xi_{(v_l,v_{l+1})}(\tilde G) ~=~ O\left(\frac{1}{1-\alpha}\right) \cdot \sum_{l \in [n-1]} \delta^\alpha(v_l,v_{l+1}) ~=~ O\left(\frac{1}{1-\alpha}\right) \cdot \omega(MST(\tilde X)).$$
By Corollary \ref{corwt},
$\tilde R \le 4 \cdot \tilde W = O(\frac{1}{1-\alpha}) \cdot \omega(MST(\tilde X))$. 
Finally, recall that the weight of the net-tree spanner for $\tilde X$ exceeds $\tilde R$ by at most a factor of 
$\eps^{-O(\frac{\ddim(X)}{\alpha})}$, and so its lightness is bounded by 
$\eps^{-O(\frac{\ddim(X)}{\alpha})} \cdot \frac{1}{1-\alpha}$.
\begin{theorem} \label{thm:basic}
Let $M=(X,\delta)$ be an $n$-point doubling metric.
For any $0 < \alpha < 1$ and   $\eps > 0$,
the lightness of the net-tree spanner for the $\alpha$-snowflake metric $(X,\delta^\alpha)$
is $\eps^{-O(\frac{\ddim(M)}{\alpha})} \cdot \frac{1}{1-\alpha}$.
\end{theorem}

\section{Applications}\label{sec:apps}
In ICALP'13  Chan {\em et al.\ } \cite{CLNS13} showed that, in any doubling metric, one can build a $(1+\epsilon)$-spanner with 
constant degree and 
with logarithmic diameter and lightness, within $O(n \log n)$ time. 
A close examination of the construction of \cite{CLNS13} shows that its lightness
is dominated by the lightness of the standard net-tree spanner.
Plugging Theorem \ref{thm:basic} in the construction of \cite{CLNS13} gives rise to a logarithmic improvement
in the lightness bound, for snowflake doubling metrics. This result is summarized in the following statement.
\begin{corollary} \label{tm1}
Let $M=(X,\delta)$ be an $n$-point doubling metric.
For any $0 < \alpha < 1$ and   $\eps > 0$,
there exists a $(1+\eps)$-spanner for the corresponding $\alpha$-snowflake metric $(X,\delta^\alpha)$
with degree $\eps^{-O(\frac{\ddim(M)}{\alpha})}$, diameter $O(\log n)$ and lightness 
$\eps^{-O(\frac{\ddim(M)}{\alpha})} \cdot \frac{1}{1-\alpha}$.
The runtime of this construction is $\eps^{-O(\frac{\ddim(M)}{\alpha})}  (n \log n)$. 
\end{corollary}
{\bf Remark.} 
For the 1-dimensional Euclidean metric $\vartheta_n$ discussed in Section \ref{sec:int},
any $(1+\eps)$-spanner with diameter $O(\log n)$ must have lightness  $\Omega(\log n)$ \cite{DES08}. 
In contrast, Corollary \ref{tm1} shows that one can get diameter $O(\log n)$ together with constant lightness
in snowflake doubling metrics.
This reveals a fundamental difference between snowflake and non-snowflake doubling metrics in the context of light spanners.

The construction of \cite{CLNS13} was extended to the fault-tolerant (FT) setting in \cite{CLNS13,Sol13b}.
Plugging Theorem \ref{thm:basic} in the FT constructions of \cite{CLNS13,Sol13b} 
gives rise to a logarithmic improvement
in the lightness bound, for snowflake doubling metrics. This result is summarized in the following statement.
\begin{corollary} \label{tm2}
Let $M=(X,\delta)$ be an $n$-point doubling metric.
For any $0 < \alpha < 1$, any $\eps > 0$ and any integer $0 \le k \le n-2$,
there exists a $k$-FT $(1+\eps)$-spanner for the corresponding $\alpha$-snowflake metric $(X,\delta^\alpha)$
with degree $\eps^{-O(\frac{\ddim(M)}{\alpha})} \cdot k$, diameter $O(\log n)$ and lightness 
$\eps^{-O(\frac{\ddim(M)}{\alpha})} \cdot \frac{1}{1-\alpha}  (k^2)$.
The runtime of this construction is $\eps^{-O(\frac{\ddim(M)}{\alpha})} (n \log n + kn)$. 
\end{corollary}

We turn to algorithms for TSP. A celebrated result of Arora \cite{A-98} states that for Euclidean TSP,
a $(1+\eps)$-approximate tour in $d$-dimensional space can be computed in time $n (\log n)^{\eps^{-O(d)}}$, and this result 
generalizes to $\ell_p$ spaces as well. Rao and Smith \cite{RS-98} utilized a light Euclidean spanner to reduce
the runtime to $2^{\eps^{-O(d)}}n + 2^{O(d)} n \log n$, and a light $\ell_p$ spanner is needed to extend
their results to $\ell_p$ space. As a consequence of Lemma \ref{lem:lp} we have:

\begin{corollary}\label{cor:lptsp}
Let $X$ be a set of $n$ $d$-dimensional vectors.
A $(1+\eps)$-approximate $\ell_p$ tour for $X$ can be computed in time
$2^{\eps^{-\tilde{O}(d)}}n + 2^{\tilde{O}(d)} n \log n$.
\end{corollary}

The best runtime for producing a $(1+\eps)$-approximate metric tour for an $n$-point metric $X$ is about 
$n^{2^{O(\ddim(X))}}$ for fixed $\eps$ \cite{BGK12}. For the $\alpha$-snowflake of $X$, 
Theorem \ref{thm:basic} provides a light spanner. Using an appropriate net hierarchy 
(see \cite{BG-13}, for example) we can utilize the machinery of Rao and Smith \cite{RS-98} and 
obtain a significantly better result.

\begin{corollary}
Let $M=(X,\delta)$ be an $n$-point doubling metric.
A $(1+\eps)$-approximate metric tour for the $\alpha$-snowflake of $M$ 
can be computed in time $2^{\frac{1}{1-\alpha} \eps^{-\tilde{O}(\ddim(M)/\alpha)}}n + 2^{\tilde{O}(\ddim(M)/\alpha)} n \log n$.
\end{corollary}

\section*{Acknowledgments}
The authors are grateful to Michael Elkin, Ofer Neiman and Michiel Smid for helpful discussions.

\bibliographystyle{alpha}
\bibliography{latex8,adi}

\end {document}